\documentclass[a4paper,11pt]{article} 
\pdfoutput=1

\usepackage{jheppub}

\usepackage{color}

\usepackage{array}

\newcolumntype{C}{>{$}c<{$}}
\newcolumntype{L}{>{$}l<{$}}
\newcolumntype{R}{>{$}r<{$}}

\renewcommand{\arraystretch}{1.1}

\toccontinuoustrue

\usepackage{tikz-cd}

\usepackage{amsrefs,amsfonts,amsthm}

\usepackage{enumerate}

\newtheorem*{theorem}{Theorem}

\newtheorem*{corollary}{Corollary}

\DeclareMathOperator{\gh}{gh}
\DeclareMathOperator{\pa}{p}
\DeclareMathOperator{\grad}{grad}

\newcommand{\A}{\mathcal{A}}
\newcommand{\CE}{\mathcal{E}}
\newcommand{\CO}{\mathcal{O}}
\newcommand{\F}{\mathcal{F}}
\newcommand{\I}{\mathcal{I}}
\newcommand{\CC}{\mathbb{C}}
\newcommand{\Z}{\mathbb{Z}}
\newcommand{\R}{\mathcal{R}}
\newcommand{\RR}{\mathbb{R}}
\newcommand{\p}{\partial}
\newcommand{\half}{\tfrac{1}{2}}

\newcommand{\Q}{\mathsf{Q}}

\newcommand{\Om}{\Omega}

\renewcommand{\[}{\{\!\{}

\renewcommand{\P}{\mathsf{P}}
\newcommand{\X}{\mathsf{X}}
\newcommand{\E}{\mathsf{E}}
\newcommand{\C}{\mathsf{C}}

\newcommand{\s}{\mathsf{s}}
\newcommand{\eps}{\varepsilon}
\newcommand{\g}{\mathsf{g}}

\title{The Batalin-Vilkovisky cohomology of the spinning particle}

\author{Ezra Getzler}

\affiliation{Department of Mathematics, Northwestern University,
  Evanston, Illinois, 60657 USA}

\emailAdd{getzler@northwestern.edu}


\abstract{We show that the axiom of Felder and Kazhdan \citep{FK} on
  the vanishing of the cohomology groups in negative degree associated
  to solutions of the classical master equation in the
  Batalin-Vilkovisky formalism is violated by the spinning particle in
  a flat background coupled to $D=1$ supergravity. In this model,
  there are nontrivial cohomology groups in all negative degrees,
  regardless of the dimension of the spacetime in which the spinning
  particle is propagating.}

\begin{document}

\maketitle
\flushbottom

\newpage

In the 1970s, the BRST formalism was introduced as a tool in the
pertubative quantization of gauge theories: it permits the use of
rather general, non-gauge invariant regularizations, while still
guaranteeing the gauge invariance of the pertubatively quantized
amplitudes. The Batalin-Vilkovisky (BV) formalism was introduced in
the 1980s to extend the BRST formalism to more general situations such
as supersymmetric theories, especially those with extended
supersymmetry, and supergravity.

The BV formalism has both classical and quantum versions. In this
paper, we restrict attention to the classical BV formalism.

We will show that in a toy model of supergravity, namely the spinning
particle, in which the worldsheet has dimension $1$, there is a series
of cohomology classes in negative degrees, violating a basic axiom for
the BV formalism which has been formulated recently by Felder and
Kazhdan \citep{FK}. The nontrivial cohomology classes which we find
are directly related to the possibility of formally inverting the
superghost of the theory, associated to local supersymmetry.

This calculation seems to indicate that the BV prescription for
associating ghost sectors to symmetries of the theory requires
modification in the presence of local supersymmetry.

\section{The Batalin-Vilkovisky formalism}

The basic characteristic of the BV formalism is the doubling of the
number of fields of the theory. Whereas in the BRST formalism, one has
a series of fields $\Phi_i$ with ghost number $\gh(\Phi_i)\ge0$, in
the BV formalism, these are supplemented by antifields $\Phi_i^+$,
with ghost number $\gh(\Phi_i^+)<0$, in such a way that
$\gh(\Phi_i)+\gh(\Phi_i^+)=-1$.

Fields with ghost number $0$ are interpreted as physical fields, while
fields with ghost number greater than $0$ are interpreted as
ghosts. (Ghosts properly speaking are fields of ghost number $1$:
fields of ghost number greater than $1$ are \emph{higher} ghosts.)
There is a similar division between those antifields with ghost number
$-1$, corresponding to physical fields, and those with ghost number
less than $-1$, corresponding to ghosts.

Besides its ghost number, each field $\Phi_i$ carries a parity
$\pa(\Phi_i)\in\{0,1\}$, defined modulo $2$, which distinguishes
between bosonic and fermionic fields. The parity of the antifield
$\Phi_i^+$ paired to a field $\Phi_i$ satisfies
$\pa(\Phi_i)=1-\pa(\Phi_i^+)$.

In this paper, we focus on the classical BV formalism with a single
independent variable $t$ (classical mechanics). Let $\p$ denote the
total derivative with respect to $t$. Denote by $\A^j$ the superspace
of all differential expressions in the fields and antifields with
$\gh(S)=j$. For example, if there is a single physical field $\phi$,
then $\A^j$ consists of all expressions in the field $\phi$ and its
derivatives $\{\p^\ell\phi\}_{\ell>0}$, and in the antifield $\phi^+$
and its derivatives $\{\p^\ell\phi^+\}_{\ell>0}$, polynomial in
$\{\p^\ell\phi\}_{\ell>0}\cup\{\p^\ell\phi^+\}_{\ell\ge0}$, and
homogenous of degree $-j$ in $\{\p^\ell\phi^+\}_{\ell\ge0}$. The sum
$\A$ of the superspaces $\A^j$ for $j\in\Z$ is a graded commutative
superalgebra.

The superspace $\F$ of functionals is the graded quotient $\A/\p\A$ of
the algebra $\A$ of currents by the subspace $\p\A$ of total
derivatives. Denote the image of $f\in\A$ in $\F$ by $\int f\,dt$.

The fields and antifields are canonical coordinates for the
Batalin-Vilkovisky antibracket, which is a Poisson bracket of degree
$1$:
\begin{equation*}
  \{ \Phi_i , \Phi_j^+ \} = - \{ \Phi_j^+ , \Phi_i \} = \delta_{ij} .
\end{equation*}
In the finite dimensional setting, the antibracket satisfies the
following equations:
\begin{description}
\item[(skew symmetry)] $\{ f , g \} = - (-1)^{(\pa(f)+1)(\pa(g)+1)} \{ g , f \}$
\item[(Jacobi)] $\{ f , \{ g , h \} \} = \{ \{ f , g \} , h \} +
  (-1)^{(\pa(f)+1)(\pa(g)+1)} \{ g , \{ f , h \} \}$
\item[(Leibniz)] $\{ f , gh \} = \{f , g \} h + (-1)^{(\pa(f)+1)\pa(g)}
  g \{ f , h \}$
\end{description}
In our setting, the space of functionals $\F$ no longer carries a
product, and we must we give up the Leibniz formula; in particular,
the antibracket is not characterized by the canonical relations
alone. Neverthless, the following antibracket makes the space of
functionals $\F$ into a graded Lie algebra:
\begin{equation*}
  { \textstyle \bigl\{ \int f\,dt , \int g\,dt \bigr\} }=
  \sum_i (-1)^{(\pa(f)+1)\pa(\Phi_i)} \int \biggl( \frac{\delta f}
  {\delta\Phi_i} \frac{\delta g}{\delta\Phi_i^+} + (-1)^{\pa(f)}
  \frac{\delta f} {\delta\Phi_i^+} \frac{\delta g}{\delta\Phi_i}
  \biggr) \, dt .
\end{equation*}
Here, $\delta/\delta\Phi_i$ and $\delta/\delta\Phi_i^+$ are the
variational derivatives, and the sum is over the fields of the
theory.

\section{The classical master equation and the Batalin-Vilkovisky
  cohomology}

The Batalin-Vilkovisky formalism for classical theory involves the
selection of a solution of the classical master equation
\begin{equation*}
  \textstyle
  \{ \int S\,dt , \int S\,dt \} = 0 ,
\end{equation*}
where $S\in\A^{0,0}$ is a differential expression in the fields and
antifields with $\gh(S)=0$ and $\pa(S)=0$. In other words, there is an
expression $\tilde{S}\in\A^{1,1}$ such that
\begin{equation*}
  \sum_i (-1)^{\pa(\Phi_i)} \frac{\delta S}{\delta\Phi_i}
  \frac{\delta S}{\delta\Phi_i^+} = \p\tilde{S} .
\end{equation*}

We may decompose $S$ into its components
\begin{equation*}
  S = S_{[0]} + S_{[1]} + \dots,
\end{equation*}
where $S_{[k]}$ is homogenous of degree $k$ in the ghosts. The BRST
formalism in its original form is the special case in which all ghost
fields have ghost number $1$, and $S_{[k]}=0$ for $k>2$.

Define an operator $\s:\F^k\to\F^{k+1}$ by the formula
\begin{equation*}
  \textstyle
  \s = \{ \int S\,dt , - \} .
\end{equation*}
By the classical master equation, this operator satisfies the equation
$\s^2=0$, and we may consider its cohomology groups
$H^*(\F,\s)$. These are called the BV cohomology associated to the
solution $S$ of the classical master equation. The graded vector space
$H^*(\F,\s)$ is a graded Lie algebra, with bracket $\{-,-\}$ of degree
$1$.

The operator $\s$ lifts to an evolutionary vector field on $\A$ of
degree $1$, characterized by the formulas
\begin{align*}
  \s\Phi_i &= (-1)^{\pa(\Phi_i)+1} \frac{\delta S}{\delta\Phi^+_i} &
  \s\Phi^+_i &=  (-1)^{\pa(\Phi_i)} \frac{\delta S}{\delta\Phi_i} .
\end{align*}
In the model considered in this paper, this vector field is
cohomological:
\begin{equation}
  \label{nilpotent}
  \s^2 = 0 .
\end{equation}
This is not the case for more typical solutions of the classical
master equation: we explain how to deal with these in a sequel to this
paper.

Since \eqref{nilpotent} holds, we may calculate the BV cohomology
groups $H^*(\F,\s)$ using the complex
\begin{equation*}
  \mathcal{V}^k = \A^k \oplus \tilde{\A}^{k+1} \, \eps ,
\end{equation*}
where
\begin{equation*}
  \tilde{\A}^k =
  \begin{cases}
    \A^0/\C , & k=0 , \\
    \A^k , & k\ne0 .
  \end{cases}
\end{equation*}
The differential on this complex is given by the formula
\begin{equation*}
  f + g\,\eps \mapsto \bigl( \s f + (-1)^{\pa(g)} \, \p g \bigr) +
  \s g\,\eps .
\end{equation*}
In this way, we obtain a long exact sequence
\begin{equation}
  \label{long}
  \begin{tikzcd}
    \cdots \arrow{r} & H^{-1}(\A,\s) \arrow{r}{\p} & H^{-1}(\A,\s)
    \arrow{r} & H^{-1}(\F,\s)  \arrow[out=-5,in=170]{lld} & \\
    & H^0(\A/\C,\s) \arrow{r}{\p} & H^0(\A,\s) \arrow{r} & H^0(\F,\s)
    \arrow[out=-5,in=170]{lld} & \\
    & H^1(\A,\s) \arrow{r}{\p} & H^1(\A,\s) \arrow{r} & H^1(\F,\s)
    \arrow{r} & \cdots
  \end{tikzcd}
\end{equation}

\section{The axiom of Felder and Kazhdan}

\label{axiom}

The above formalism is extremely general. In an attempt to constrain
the possible models to a class exhibiting the features of the
classical field theories which are of interest in theoretical physics,
Felder and Kazhdan \citep{FK} distilled from the work of Batalin and
Vilkovisky an axiom for solutions of the classical master
equation. Following their paper, we start by formulating the axiom in
the finite-dimensional setting, in which fields are replaced by
coordinates in a finite-dimensional graded supermanifold. (They
restrict attention to graded manifolds, but for theories incorporating
fermions, one must consider graded supermanifolds.)

Let $\I\subset\A$ be the ideal of $\A$ generated by the coordinates
with positive ghost number, and let $\A/\I$ be the quotient graded
superalgebra. The ideal $\I$ is closed under the action of $\s$, so the
differential descends to the quotient graded superalgebra $\A/\I$, and
we may define the cohomology $H^*(\A/\I,\s)$, which is again a graded
superalgebra. There is a quotient map
\begin{equation*}
  H^*(\A,\s) \to H^*(\A/\I,\s) ,
\end{equation*}
which is in general neither surjective nor injective.

Following Batalin and Vilkovisky, Felder and Kazhdan propose the axiom
that the cohomology of $\A/\I$ should vanish in negative degrees: if
$k>0$,
\begin{equation*}
  H^{-k}(\A/\I,\s) = 0 .
\end{equation*}
The extension of this axiom to field theory is extremely powerful: for
example, starting with the classical action for the Yang-Mills theory,
it leads to the discovery of the gauge symmetry and its ghosts: the
antifields of the ghosts are needed in order to kill the cohomology
classes $\int (\epsilon,d_AA^+)$ of degree $-1$ whose closure is a
consequence of the second Bianchi identity for the curvature. (Here,
$\epsilon$ is a section of the adjoint bundle $P\times_G\g$ associated
to the principal bundle $P$, and the antifield $A^+$ is a three-form
with values in $P\times_G\g$.)

By a spectral sequence argument, Felder and Kazdhan prove that their
axiom implies the vanishing of the cohomology of $\A$ in negative
degree: if $k>0$,
\begin{equation*}
  H^{-k}(\A,\s) = 0 .
\end{equation*}
Filtering $\A$ by powers of $\I$, one obtains a spectral sequence with
$E_2$-term
\begin{equation*}
  \bigoplus_{\ell=0}^\infty H^*(\I^\ell/\I^{\ell+1},\s) \Longrightarrow
  H^*(\A,\s) .
\end{equation*}
Each of the complexes $\I^\ell/\I^{\ell+1}$ is a free module
over $\A/\I$, with basis in strictly positive degrees: hence, if
$H^{-k}(\A/\I,\s)$ vanishes for all $k>0$, it follows that
$H^{-k}(\A,\s)$ vanishes for all $k>0$.

It is not hard to formulate the extension of the axiom of Felder and
Kazhdan to the setting of classical field theory: here, we restrict
attention to the case of a single independent variable $t$. Let
$\I\subset\A$ be the ideal of $\A$ generated by the fields with
positive ghost number and their partial derivatives, and let $\A/\I$
be the quotient graded superalgebra. There is a quotient complex
\begin{equation*}
  \F \to \A/(\p\A+\I) ,
\end{equation*}
and a quotient map
\begin{equation*}
  H^*(\F,\s) \to H^*(\A/(\p\A+\I),\s) ,
\end{equation*}
which is in general neither surjective nor injective.

The natural extension of the axiom of Felder and Kazhdan is that the
cohomology of $\A/(\p\A+\I)$ should vanish in negative degrees: if
$k>0$,
\begin{equation*}
  H^{-k}(\A/(\p\A+\I),\s) = 0 .
\end{equation*}
The spectral sequence argument of Felder and Kazdhan extends to this
setting, and shows that their axiom implies the vanishing of the
cohomology of $\F$ in degree less than $-1$: if $k>1$,
\begin{equation*}
  H^{-k}(\F,\s) = 0 .
\end{equation*}
Barnich et al.\ have shown that this holds for the pure Yang-Mills
theory on $\RR^4$, with arbitrary semisimple gauge group \citep{BBH}
(see also Costello \citep[Chapter 6, Section 5]{Costello}).
 
We will show that there exists a field theory in which the axiom of
Felder and Kazhdan fails, in the sense that $H^{-k}(\F,\s)$ is nonzero
for all $k>1$.

\section{Review of the spinning particle}

We now investigate the axiom of Felder and Kazhdan in a toy model of
supergravity, the spinning particle~\citep{spinning}. Consider the
vector space $\RR^d$ with non-degenerate inner product
\begin{equation*}
  ( e_\mu , e_\nu ) = \eta_{\mu\nu} .
\end{equation*}
The spinning particle has physical fields $x^\mu=x^\mu(t)$ and
$\theta^\mu=\theta^\mu(t)$ of parity $0$ and $1$ respectively, and action
\begin{equation*}
  S = \half \eta_{\mu\nu} \bigl( \p x^\mu \p x^\nu - \theta^\mu
  \p\theta^\nu \bigr) ,
\end{equation*}
where $\p x^\mu$ and $\p\theta^\mu$ are the derivatives of the fields
$x^\mu$ and $\theta^\mu$ with respect to the independent variable $t$
parametrizing the time-line of the particle. We denote by $\CO$ the
algebra of functions in the variables $x^\mu$. There is a lot of
flexibility in the choice of this algebra: we may take polynomials in
the variables $x^\mu$, infinitely-differentiable functions, analytic
functions, or even power series. All of our results will be
independent of this choice. Thus, $\A$ is the graded polynomial
algebra over $\CO$ generated by the remaining variables of the theory,
namely $\{\p^\ell x^\mu\}_{\ell>0}\cup
\{\p^\ell\theta,\p^\ell x^+_\mu,\p^\ell\theta^+_\mu\}_{\ell\ge0}$.

We prefer to work in a first-order formulation of this theory, which
has an additional physical field $p_\mu$, with even parity, and the
modified action
\begin{equation*}
  S = p_\mu \p x^\mu - \half \eta_{\mu\nu} \theta^\mu \p\theta^\nu -
  \half \eta^{\mu\nu} p_\mu p_\nu .
\end{equation*}
The differential $s$ on the fields and antifields of the theory is
given by the formulas
\begin{align*}
  \s x^+_\mu &= - \p p_\mu &
  \s\theta^+_\mu &= \eta_{\mu\nu} \p \theta^\nu &
  \s p^{+\mu} &= \p x^\mu - \eta^{\mu\nu} p_\nu \\
  \s x^\mu &= 0 & \s\theta^\mu &= 0 & \s p_\mu &= 0
\end{align*}
This differential is an example of a Koszul complex, and the
cohomology $H^{-k}(\A/\I,\s)$ vanishes in negative degree, and in
degree $0$ is the graded polynomial ring in the functionals
$\int x^\mu\,dt$, $\int \theta^\mu\,dt$ and $\int p_\mu\,dt$. In
particular, the axiom of Felder and Kazhdan is seen to hold.

In order to have a theory with local reparametrization invariance and
local supersymmetry, we may couple the spinning particle to
supergravity. Of course, the gravitational field in worldsheet
dimension $1$ has no dynamical content: but we will see that the
ghosts for the local supersymmetry of the theory considerably
complicate matters.

The supergravity multiplet consists of a pair of physical fields $e$
and $\psi$, of parity $\pa(e)=0$ and $\pa(\psi)=1$. These fields,
respectively a $1$-form and a function, may be identified with the
graviton and the gravitino of $D=1$ supergravity. The new action is
\begin{equation*}
  S_{[0]} = p_\mu \p x^\mu - \half \eta_{\mu\nu} \theta^\mu \p\theta^\nu -
  \half e \eta^{\mu\nu} p_\mu p_\nu + \psi p_\mu \theta^\mu .
\end{equation*}
The differential is now
\begin{align*}
  & & \s_{[0]}e^+ &= - \half \eta^{\mu\nu} p_\mu p_\nu
  & \s_{[0]}\psi^+ &= - p_\mu \theta^\mu \\
  \s_{[0]}x^+_\mu &= - \p p_\mu
  & \s_{[0]}\theta^+_\mu &= \eta_{\mu\nu} \p \theta^\nu + \psi p_\mu
  & \s_{[0]}p^{+\mu} &= \p x^\mu - e \eta^{\mu\nu} p_\nu + \psi \theta^\mu
\end{align*}
The variation $\s_{[0]}e^+$ is familiar as the stress-energy
tensor. The local gauge symmetries of this model correspond to
cohomology classes of $s_{[0]}$ at ghost number $-1$:
\begin{align*}
  \s_{[0]}\bigl( \p e^+ - \eta^{\mu\nu} x^+_\mu p_\nu \bigr) &= 0 &
  \s_{[0]}\bigl( \p \psi^+ + \eta^{\mu\nu} \theta^+_\mu p_\nu - x^+_\mu \theta^\mu
  + 2 e^+ \psi \bigr) &= 0 .
\end{align*}
These cohomology classes are killed by the introduction of ghosts $c$
and $\gamma$, with ghost-number $1$ and parity $\pa(c)=1$ and
$\pa(\gamma)=0$, and the corresponding antifields, and the addition to
the action of the term
\begin{equation*}
  S_{[1]} = \bigl( \p e^+ - \eta^{\mu\nu} x^+_\mu p_\nu \bigr) c
  + \bigl( \p \psi^+ + \eta^{\mu\nu} \theta^+_\mu p_\nu -
  x^+_\mu \theta^\mu + 2 e^+ \psi \bigr) \gamma .
\end{equation*}
This adds the following terms to the differential:
\begin{align*}
  \s_{[1]}c^+ &= \p e^+ - \eta^{\mu\nu} x^+_\mu p_\nu & \s_{[1]}\gamma^+ &=
  \p \psi^+ + \eta^{\mu\nu} \theta^+_\mu p_\nu - x^+_\mu \theta^\mu + 2
  e^+ \psi \\
  \s_{[1]}\psi^+ &= 2 e^+ \gamma & & \\
  \s_{[1]}\theta^+_\mu &= - x^+_\mu \gamma &
  \s_{[1]}p^{+\mu} &= - \eta^{\mu\nu} x^+_\nu c + \eta^{\mu\nu}
                    \theta^+_\nu \gamma \\
  \s_{[1]}x^\mu &= - \eta^{\mu\nu} p_\nu c - \theta^\mu \gamma &
  \s_{[1]}\theta^\mu &= - \eta^{\mu\nu} p_\nu \gamma \\
  \s_{[1]}e &= - \p c + 2 \psi \gamma & \s_{[1]}\psi &= \p\gamma
\end{align*}
The variation $\s_{[1]}c^+$ reflects the conservation of the
stress-energy tensor. The definition of the action is completed by the
addition of the term
\begin{equation*}
  S_{[2]} = - c^+ \gamma^2
\end{equation*}
yielding a solution $S=S_{[0]}+S_{[1]}+S_{[2]}$ of the classical
master equation
\begin{equation*}
  \textstyle
  \{ \int S\,dt , \int S\,dt \}=0 .
\end{equation*}
This adds the following terms to the differential, rendering it
nilpotent:
\begin{align*}
  \s_{[2]}\gamma^+ &= - 2 c^+ \gamma & \s_{[2]}c &= \gamma^2 .
\end{align*}
In this model, $\A$ is the graded polynomial algebra over $\CO$
generated by the remaining variables of the theory, namely
\begin{align*}
  & \{\p^\ell x^\mu\}_{\ell>0}
    \cup \{\p^\ell\theta,\p^\ell p_\mu,\p^\ell e,\p^\ell\psi,\p^\ell c,
    \p^\ell\gamma\}_{\ell\ge0}\\
  \mbox{} \cup
  & \{\p^\ell x^+_\mu,\p^\ell\theta^+_\mu,\p^\ell p^{+\mu},\p^\ell
    e^+,\p^\ell\psi^+,\p^\ell c^+,\p^\ell\gamma^+\}_{\ell\ge0} .
\end{align*}

As it happens, the action $S$ satisfies not just the classical master
equation but also the quantum master equation, since for each field
$\Phi$ of the model, we have
\begin{equation*}
  \frac{\p^2S}{\p\Phi\p\Phi^+} = 0 .
\end{equation*}
This is characteristic of quantum mechanics, where one does not meet
with anomalies in the course of quantization.

\section{General covariance of the spinning particle}

The Lie algebra of vector fields on the real line acts on the spinning
particle by reparametrization of the independent variable
$t$. Reflecting this, this Lie algebra acts on $\A$ by evolutionary
vector fields: the vector field $-\xi\p/\p t$ on the real line acts by
the evolutionary vector field
\begin{align*}
  T(\xi)\Phi
  &= \xi \left( \p x^\mu \frac{\p}{\p x^\mu} + \p\theta^\mu
    \frac{\p}{\p\theta^\mu} + \p p_\mu \frac{\p}{\p p_\mu} + \p
    x^+_\mu \frac{\p}{\p x^+_\mu} + \p\theta^+_\mu
    \frac{\p}{\p\theta^+_\mu} + \p p^{+\mu} \frac{\p}{\p p^{+\mu}}
    \right. \\
  &\quad \left. + \p e \frac{\p}{\p e} + \p\psi \frac{\p}{\p\psi} 
    + \p e^+ \frac{\p}{\p e^+} + \p\psi^+ \frac{\p}{\p\psi^+} + \p c
    \frac{\p}{\p c} + \p\gamma \frac{\p}{\p\gamma}
    + \p c^+ \frac{\p}{\p c^+} + \p\gamma^+ \frac{\p}{\p\gamma^+} \right) \\
  &- \p\xi \left( e \frac{\p}{\p e} + \psi \frac{\p}{\p\psi}
    + x^+_\mu \frac{\p}{\p x^+_\mu} + \theta^+_\mu
    \frac{\p}{\p\theta^+_\mu} + p^{+\mu} \frac{\p}{\p p^{+\mu}} + \p
    c^+ \frac{\p}{\p c^+} + \p\gamma^+ \frac{\p}{\p\gamma^+} \right) .
\end{align*}
(We may think of $\xi$ either as a function of the independent
variable $t$, in which case $\p\xi=\xi'$, or, more in the spirit of
this article, as an auxiliary bosonic field.) In particular,
$T(1)=\p$.

Consider the expression
\begin{equation*}
  G = x^+_\mu p^{+\mu} - \half \eta^{\mu\nu}\theta^+_\mu\theta^+_\nu +
  c^+e + \gamma^+\psi .
\end{equation*}
The general covariance of the spinning particle is expressed through
the formula
\begin{equation*}
  \{ \{S,\xi G\} , f \} = T(\xi)f ,
\end{equation*}
which follows from the identity
\begin{equation*}
  \{S,\xi G\} = - \xi \bigl( x^+_\mu \p x^\mu + \theta^+_\mu
  \p\theta^\mu + p^{+\mu} \p p_\mu - \p e^+ e + c^+ \p c - \p\psi^+
  \psi + \gamma^+ \p \gamma \bigr) .
\end{equation*}
As a special case, we have
\begin{equation*}
  \{ \{S, G\} , f \} = \p f .
\end{equation*}
This corresponds to the commutator
\begin{equation}
  \label{homotopy}
  [ \s , \g ] = \p
\end{equation}
among evolutionary vector fields, where $\g$ is the vector field
$\g(f)=\{G,f\}$ associated to $G$:
\begin{equation}
  \label{topological}
  \g = p^{+\mu} \frac{\p}{\p x^\mu} - x_\mu^+ \frac{\p}{\p p_\mu} +
  \eta^{\mu\nu} \theta^+_\mu \frac{\p}{\p\theta^\mu} + c^+
  \frac{\p}{\p e^+} - e \frac{\p}{\p c} + \gamma^+ \frac{\p}{\p\psi^+}
  + \psi \frac{\p}{\p\gamma} .
\end{equation}
We see that the total derivative $\p$ acts trivially on the cohomology
$H^*(\A,\s)$.

\section{The Batalin-Vilkovisky cohomology of the spinning particle:
  $d=0$}

In this section, we calculate the BV cohomology of the spinning
particle when the matter fields $x^\mu,\theta^\mu,p_\mu$ are absent,
that is, when the dimension $d$ of the target space $\RR^d$ is
zero. We will see that the cohomology groups are nontrivial in all
negative degrees.

Consider the following elements of the localization $\A_\gamma$ of the
complex $\A$ obtained by inversion of $\gamma$:
\begin{align*}
  A_k &= (\psi^+)^{k+1} c \gamma^{-1} \in \A_\gamma^{-k-1} , & & k\ge-1 , \\
  B_k &= \frac{1}{k} (\psi^+)^k \gamma^{-1} \in \A_\gamma^{-k-1} , & & k\ge1 .
\end{align*}
The BV differentials of these expressions are in $\A$ itself:
\begin{align*}
  \alpha_k &= 2(k+1) (\psi^+)^k e^+ c + (\psi^+)^{k+1} \gamma \in
          \A^{-k} , & & k\ge-1 , \\
  \beta_k &= (\psi^+)^{k-1} e^+ \in \A^{-k} , & & k\ge1 .  
\end{align*}
It follows that $\alpha_k$ and $\beta_k$ are cocycles: it is not
difficult to see that they are not coboundaries in $\A$ itself.

\begin{theorem}
  \label{nomatter}
  The cohomology group $H^{-k}(\A,\s)$ is spanned by the cocycles
  $\alpha_k$ and $\beta_k$, and the cocycle $1$ in degree $0$.
\end{theorem}
\begin{proof}
  The differential $\s$ is a quadratic perturbation of a Koszul
  differential, and its cohomology may be calculated by a spectral
  sequence with $E_0$ equal to $\A$ with the differential $\s_0$
  obtained by retaining only linear terms in the formula for $s$:
  \begin{align*}
    \s_0 c^+ &= \p e^+ & \s_0 \gamma^+ &= \p \psi^+ &
    \s_0 e &= - \p c & \s_0 \psi &= \p\gamma .
  \end{align*}
  The cohomology $E_1$ of the differential $\s_0$ is a graded
  polynomial ring in generators $\E^+=\int e^+\,dt$ and $\Psi^+=\int
  \psi^+\,dt$, in degree $-1$, and $\C=\int c\,dt$ and
  $\Gamma=\int\gamma\,dt$, in degree $1$.

  The differential $\s_1$ on $E_1$ is given by the formulas
  \begin{align*}
    \s_1\Psi^+ &= 2 \E^+ \Gamma & \s_1\C &= \Gamma^2 .
  \end{align*}
  A cochain $z$ of degree $-k$ for $k>0$ has the following general
  form:
  \begin{multline*}
    z = u_k (\Psi^+)^k + v_k (\Psi^+)^{k-1} \E^+ \\
    + \sum_{j=1}^\infty \Bigl( (\Psi^+)^{k+j} ( u_{k+j} \Gamma^j +
    U_{k+j} \C \Gamma^{j-1} ) + (\Psi^+)^{k+j-1} \E^+ ( v_{k+j}
    \Gamma^j + V_{k+j} \C \Gamma^{j-1} ) \Bigr) .
  \end{multline*}
  Setting $\s_1z=0$, we see that $U_{k+j}=V_{k+j}-2(k+j)u_{k+j}=0$,
  and that
  \begin{equation*}
    z = u_{k+1}\alpha_k + v_k\beta_k + s_1 \Biggl(
    \sum_{j=2}^\infty u_{k+j} (\Psi^+)^{k+j} \C \Gamma^{j-2} +
    \sum_{j=1}^\infty \frac{v_{k+j}}{2(k+j)} (\Psi^+)^{k+j} \Gamma^{j-1}
    \Biggr) .
  \end{equation*}
  This shows that if $k>0$, $H^{-k}(\A,\s)$ is spanned by $\alpha_k$
  and $\beta_k$.

  A cochain $z$ of degree $0$ has the following general form:
  \begin{equation*}
    z = u_0 + \sum_{j=1}^\infty \Bigl( (\Psi^+)^j ( u_j \Gamma^j + U_j
    \C \Gamma^{j-1} ) + (\Psi^+)^{j-1} \E^+ ( v_j \Gamma^j + V_j
    \C \Gamma^{j-1} )\Bigr) .
  \end{equation*}
  Setting $\s_1z=0$, we see that
  \begin{equation*}
    z = u_0 + u_1 \alpha_0 + s_1 \Biggl( - \sum_{j=2}^\infty u_j
    (\Psi^+)^k \C \Gamma^{j-2} + \sum_{j=1}^\infty \frac{v_j}{2k} 
    (\Psi^+)^k \Gamma^{j-1} \Biggr) .
  \end{equation*}
  This shows that $H^0(\A,\s)$ is spanned by $1$ and
  $\alpha_0$. Similar calculations show that $H^1(\A,\s)$ is spanned
  by $\alpha_{-1}$, and that $H^\ell(\A,\s)$ vanishes for $\ell>1$.
\end{proof}

Let $\g$ be the vector field \eqref{topological}, which in the
present setting equals
\begin{equation*}
  \g = c^+ \frac{\p}{\p e^+} - e \frac{\p}{\p c} + \gamma^+
  \frac{\p}{\p\psi^+} + \psi \frac{\p}{\p\gamma}
\end{equation*}
and consider the transgressions of the cocycles $\alpha_k$ and $\beta_k$:
\begin{align*}
  \tilde{\alpha}_k &= \g(\alpha_{k-1}) \in \A^{-k}, & & k\ge0 , \\
  \tilde{\beta}_k &= \g(\beta_{k-1}) \in \A^{-k} , & & k\ge2 .
\end{align*}

\begin{corollary}
  The cohomology group $H^{-k}(\F,\s)$ is spanned by the cocycles
  $\int \alpha_k \,dt$, $\int \beta_k \,dt$, $\int \tilde\alpha_k\,dt$
  and $\int \tilde\beta_k\,dt$, and the cocycle $\int 1\,dt$ in degree
  $0$.
\end{corollary}
\begin{proof}
  We consider the long exact sequence \eqref{long}. By
  \eqref{homotopy}, the morphisms
  \begin{equation*}
    \p : H^{-k}(\A,\s) \to H^{-k}(\A,\s)
  \end{equation*}
  vanish. This implies that the classes $\int\alpha_k\,dt$ and
  $\int\beta_k\,dt$ and their transgressions
  $\int\tilde{\alpha}{}_k\,dt$ and $\int\tilde{\beta}{}_k\,dt$ span
  $H^{-k}(\F,\s)$, together with $\int 1\,dt$ in degree $0$.
\end{proof}

It turns out that the graded Lie bracket on $H^*(\F,\s)$ is nontrivial.
\begin{theorem}
  We have the following nonzero graded Lie brackets in $H^*(\F,\s)$:
  \begin{align*}
    \{ \tilde\alpha_k , \alpha_\ell \} &= (\ell-k) \alpha_{k+\ell} &
    \{ \tilde\alpha_k , \tilde\alpha_\ell \} &= (\ell-k) \tilde\alpha_{k+\ell} \\
    \{ \tilde\alpha_k , \beta_\ell \} &= (2k+\ell+1) \beta_{k+\ell} &
    \{ \tilde\alpha_k , \tilde\beta_\ell \} &= (2k+\ell+1) \tilde\beta_{k+\ell}
  \end{align*}
  All other brackets are zero.
\end{theorem}
\begin{proof}
  By \eqref{homotopy}, we see that
  \begin{equation*}
    \{ \tilde\alpha_k , \alpha_\ell \} = - \s \{ \g(A_k) , \alpha_\ell
    \} + \p \{ A_k , \alpha_\ell \} = - \s \{ \g(A_k) , \alpha_\ell \} ,
  \end{equation*}
  since $\{ A_k , \alpha_\ell \}=0$. We have
  \begin{align*}
    \{ \g(A_k) , \alpha_\ell \}
    &= \mbox{} \{ (k+1) (\psi^+)^k \gamma^+ c \gamma^{-1} -
      (\psi^+)^{k+1} e \gamma^{-1} + (\psi^+)^{k+1} c \psi \gamma^{-2}
      , \alpha_\ell \} \\
    &= (k+1) A_{k+\ell+1} - 2 (\ell+1) A_{k+\ell+1} + (\ell+1)
      A_{k+\ell+1} \\
    &= (k-\ell) A_{k+\ell} ,
  \end{align*}
  and the formula for $\{\tilde\alpha_k,\alpha_\ell\}$
  follows. Applying the vector field $\g$ to both sides of the formula
  for $\{\tilde\alpha_k,\alpha_\ell\}$, we obtain the formula for
  $\{\tilde\alpha_k,\tilde\alpha_\ell\}$.

  Similarly,
  $\{ \tilde\alpha_k , \beta_\ell \} = - \s \{ \g A_k , \beta_\ell
  \}$. Since
  \begin{align*}
    \{ \g(A_k) , \beta_\ell \}
    &= \mbox{} \{ (k+1) (\psi^+)^k \gamma^+ c \gamma^{-1} -
      (\psi^+)^{k+1} e \gamma^{-1} + (\psi^+)^{k+1} c \psi \gamma^{-2}
      , \beta_\ell \} \\
    &= - (\psi^+)^{k+\ell} \gamma^{-1} - (\ell-1)
      (\psi^+)^{k+\ell-1}e^+c\gamma^{-2} ,
  \end{align*}
  the formula for $\{\tilde\alpha_k,\beta_\ell\}$ follows. Finally,
  applying the vector field $\g$ to both sides of the formula for
  $\{\tilde\alpha_k,\beta_\ell\}$, we obtain the formula for
  $\{\tilde\alpha_k,\tilde\beta_\ell\}$. It is easily seen from the
  explicit formulas that all of the remaining brackets vanish.
\end{proof}

\section{The Batalin-Vilkovisky cohomology of the spinning particle:
  $d>0$}

The cohomology classes $\alpha_k$ and $\beta_k$ which we constructed
in the previous section have generalizations for the spinning particle
in a positive-dimensional target, where $d>0$. Here, we discuss the
case of a flat target: we generalize our results to targets with
nontrivial (pseudo-)Riemannian metric, and background magnetic field,
in a forthcoming paper.

Although the formulas for the cohomology classes of negative degree
are quite complicated, they all involve an analogue of the volume form
for $\RR^d$ constructed from the field $\theta^\mu$:
\begin{equation*}
  \Om = \theta^1 \dots \theta^d .
\end{equation*}
Since $(\Om)^2=0$ when $d>0$, it turns out that the Lie bracket on the
cohomology $H^*(\F,\s)$ is simpler in the general case.

Let $\iota_\mu=\p/\p\theta^\mu$. If $v$ is a vector with components
$v_\mu$, define
\begin{equation*}
  \iota(v) = \eta^{\mu\nu} v_\mu \iota_\nu
\end{equation*}
In particular, $[s,\iota(v)]=\iota(sv)$. Note that
$\s\Om=\iota(p)\Om\gamma$. If $f\in\CO$, denote by $\grad f$ the
vector with components $\p f/\p x^\mu$.

Given a function $f$ of the coordinates $x^\mu$ and $k\ge0$, consider
the following elements of $\A_\gamma^{-k-1}$:
\begin{align*}
  A_k(f) &= (\psi^+)^{k+1} c f \Om \gamma^{-1} , \\
  Z_k(f) &= (k+1) (\psi^+)^k f \Om \gamma^{-1} +
           (\psi^+)^{k+1} c \iota(\grad f) \Om \gamma^{-1} .
\end{align*}
The BV differentials of these expressions are in $\A^{-k}$:
\begin{align*}
  \alpha_k(f) &= \s(A_k(f)) , & \zeta_k(f) &= \s(Z_k(f)) .
\end{align*}
It follows that they are are cocycles in $\A$: we will see that they
represent nontrivial cohomology classes.

There are also cocycles in degrees $0$ and $1$ which have no analogue
in the case where $d=0$. Let $\R$ be the quotient of the differential
graded superalgebra $\A$ by the differential ideal generated by the
fields
\begin{equation*}
  \{e,\psi,c\} \cup \{
  x^+_\mu,\theta^+_\mu,p^{+\mu},e^+,\chi^+,c^+,\gamma^+\}
\end{equation*}
Denote by $\P_\mu$, $\Theta^\mu$, $\X^\mu$ and $\Gamma$ the zero-modes
$\int p_\mu\,dt$, $\int\theta^\mu\,dt$, $\int x^\mu\,dt$ and
$\int\gamma\,dt$ respectively. Then $\R$ is the graded superalgebra
\begin{equation*}
  \CO[\Theta^\mu,\P_\mu,\Gamma] /
  \bigl(\P_\mu\Theta^\mu,\zeta^{\mu\nu}\P_\mu\P_\nu,\Gamma^2 \bigr)
\end{equation*}
with differential $\Q\Gamma$, where $\Q$ is the differential operator
\begin{equation*}
  \Q = \eta^{\mu\nu} \P_\mu \frac{\p}{\p\Theta^\nu} + \Theta^\mu
  \frac{\p}{\p\X^\mu} .
\end{equation*}
We may think of $\R$ as the ring of functions whose spectrum is a
supersymmetric analogue of the light-cone. We will denote
$\Theta^1\dots\Theta^d\in\R^0$ by the same letter $\Om$ as in $\A^0$.

There is an embedding $\xi$ of $H^*(\R)$ into $H^*(\A,\s)$, defined by
mapping a function $u$ of the variables $\{\X^\mu,\Theta^\mu,\P_\mu\}$
to the corresponding function $\xi^0(u)$ of the variables
$\{x^\mu,\theta^\mu,p_\mu\}$ in $\A$. If $\Q u=0$, it follows that
$\xi^0(u)$ is a cocycle.

Since $(P_\mu\Theta^\mu)\xi^0(u)=-\s(\psi^+\xi^0(u))$ and
$(\eta^{\mu\nu} P_\mu P_\nu)\xi^0(u)=-2\s(e^+\xi^0(u))$ are
coboundaries, we see that $\xi^0$ induces a map from $H^0(\R)$ to
$H^0(\A,\s)$. Note that $\xi^0(\iota(P)\Om)=-\zeta_0(1)$.

We may also map a function $v$ of the variables
$\{\X^\mu,\theta^\mu,\P_\mu\}$ to the $1$-cocycle
\begin{equation*}
  \xi^1(v) = \gamma v + c \Q v .
\end{equation*}
By the formulas
\begin{equation*}
  \xi^1(\Q v) = - \s(\xi^1(v)) ,
\end{equation*}
and
\begin{align*}
  \xi^1\bigl( \P_\mu\Theta^\mu v \bigr)
  &= \s \bigl( 2e^+c \xi^0(v) - \psi^+ \xi^1(v) \bigr) &
  \xi^1\bigl( \eta^{\mu\nu}\P_\mu \P_\nu v \bigr)
  &= - 2 \, s \bigl( e^+ \xi^1(v) \bigr) ,
\end{align*}
we see that $\xi^1$ induces a map from $H^1(\R)$ to $H^1(\A,\s)$.

\begin{theorem}
  \begin{equation*}
    H^{-k}(\A,\s) =
    \begin{cases}
      \Bigl\{ \int \bigl( \alpha_k(f) + \zeta_k(g) \bigr) \,dt \,
      \Big| \, f,g \in\CO \Bigr\} & k>0 , \\[10pt]
      \Bigl\{ \int \bigl( \xi^0(u) + \alpha_0(f) + \zeta_0(g) \bigr)
      \,dt\,\Big|\, u\in H^0(\R) , f \in \CO , g \in\CO/\CC \Bigr\} &
      k=0 , \\[10pt]
      \Bigl\{\int\xi^1(v)\,dt\,\Big|\, v\in H^1(\R) \Bigr\} & k=-1 ,
      \\[10pt]
      0 & k<-1 .
    \end{cases}
  \end{equation*}
\end{theorem}
\begin{proof}
  The proof uses a filtration on $\A$, defined by assigning to the
  fields the following filtration degrees:
  \begin{center}
    \begin{tabular}{|C|C|C|} \hline
      \hbox{\ $\Phi$\ } &
      \hbox{\ $\deg(\Phi)$\ } &
      \hbox{$\deg(\Phi^+)$} \\ \hline
      x & 0 & 2\sigma \\
      \theta & \sigma & \sigma \\
      p & 2\sigma & 0 \\
      e & 2-2\sigma & 4\sigma-1 \\
      \psi & 2-\sigma & 3\sigma-1 \\
      c & 2-2\sigma & 4\sigma-1 \\
      \gamma & 2-\sigma & 3\sigma-1 \\ \hline
    \end{tabular}
  \end{center}
  If $\sigma$ is chosen in the interval $\frac13\le\sigma\le1$, the
  filtration degrees are all non-negative, proving that the spectral
  sequence converges in this case. We will see that in fact, the
  spectral sequence converges for all values of $\sigma$. This will be
  convenient in the sequel to this paper, where we consider a target
  with non-trivial magnetic field: to handle this case, we will need
  to choose $\sigma$ negative.

  The filtration degrees of the terms of the BV differential $\s$ are
  independent of $\sigma$, and lie between $0$ and $2$. The
  differential $\s_0$ on $E_0$ is the following Koszul differential:
  \begin{align*}
    & & \s_0 c^+ &= \p e^+ & \s_0 \gamma^+ &= \p \psi^+ \\
    \s_0 x^+_\mu &= - \p p_\mu
    & \s_0 \theta^+_\mu &= \eta_{\mu\nu} \p \theta^\nu
    & \s_0 p^{+\mu} &= \p x^\mu
    & \s_0 e^+ &= 0 & \s_0 \psi^+ &= 0 \\
    \s_0 x^\mu &= 0 & \s_0 \theta^\mu &= 0 & \s_0 p_\mu &= 0
    & \s_0 e &= - \p c & \s_0 \psi &= \p\gamma \\
    & & \s_0 c &= 0 & \s_0 \gamma &= 0
  \end{align*}
  We see that the cohomology $E_1$ of the differential $\s_0$, is a
  graded polynomial ring in the following generators:
  \begin{equation*}
    \renewcommand{\arraystretch}{1.4}
    \begin{tabular}{|r|c|} \hline
      $\gh$ & \textup{generators} \\ \hline
      $-1$ & $\E^+=\int e^+\,dt$, $\Psi^+=\int \psi^+\,dt$ \\
      $0$ & $\X^\mu=\int x^\mu\,dt$, $\Theta^\mu=\int \theta^\mu\,dt$,
            $\P_\mu=\int p_\mu\,dt$ \\
      $1$ & $\C=\int c\,dt$, $\Gamma=\int\gamma\,dt$ \\
      \hline
    \end{tabular}
  \end{equation*}
  The differential $\s_1$ on $E_1$ is given by the formulas
  \begin{align*}
    \s_1\E^+ &= - \half \eta^{\mu\nu} \P_\mu \P_\nu &
    \s_1\Psi^+ &= - \P_\mu \Theta^\mu .
  \end{align*}
  A cochain in $E_1$ has the general form
  \begin{equation*}
    z = \sum_{j>0} (\Psi^+)^{j-1} \E^+ a_j + \sum_{j\ge0} (\Psi^+)^j b_j ,
  \end{equation*}
  where $a_j,b_j\in\CO[\Theta^\mu,\P_\mu,\C,\Gamma]$. The differential
  of $z$ equals
  \begin{equation*}
    \s_1z = \sum_{j>0} j (\Psi^+)^{j-1} \E^+ \bigl( \P_\mu\Theta^\mu
    \bigr) a_{j+1} - \sum_{j\ge0} (\Psi^+)^j \bigl( (j+1) \bigl(
    \P_\mu\Theta^\mu \bigr) b_{j+1} + \bigl( \half \eta^{\mu\nu}
    \P_\mu\P_\nu \bigr) a_{j+1} \bigr) .
  \end{equation*}
  The equations
  \begin{equation}
    \label{cocycle}
    j \bigl( \P_\mu\Theta^\mu \bigr) b_j = - \bigl( \half
    \eta^{\mu\nu} \P_\mu\P_\nu \bigr) a_j
  \end{equation}
  imply the vanishing of $\s_1z$. To see this, multiply both sides of
  the equation \eqref{cocycle} by $\P_\mu\Theta^\mu$. The left-hand
  side vanishes: since $\eta^{\mu\nu} \P_\mu\P_\nu$ is not a
  zero-divisor, it follows that
  \begin{equation}
    \label{ptheta}
    \bigl( \P_\mu\Theta^\mu \bigr) a_j = 0 .
  \end{equation}

  Associate to $f\in\CO[\C,\Gamma]$ the cocycles in $E_1$
  \begin{align*}
    A_j(f) &= 2j (\Psi^+)^{j-1} \E^+ f \, \Om - (\Psi^+)^j f
    \iota(\P)\Om , &
    B_j(f) &= (\Psi^+)^j f \, \Om .
  \end{align*}

  Since the operation of multiplication by $\P_\mu\Theta^\mu$ may be
  interpreted as the differential of a Koszul complex, it has kernel
  spanned by elements of the form $\bigl( \P_\mu\Theta^\mu \bigr) u$,
  where $u\in\CO[\Theta^\mu,\P_\mu,\C,\Gamma]$, and $f\,\Om$, where
  $f\in\CO[\C,\Gamma]$. Thus, we may decompose a solution $a_j$ of
  \eqref{ptheta} as a sum
  \begin{equation*}
    a_j = \bigl( \P_\mu\Theta^\mu \bigr) u_j + 2jf_j \, \Om ,
  \end{equation*}
  where $u_j\in\CO[\Theta^\mu,\P_\mu,\C,\Gamma]$ and
  $f_j\in\CO[\C,\Gamma]$. This allows us to rewrite the cocycle $z$ as
  \begin{equation*}
    z = b_0 + \sum_{j>0} A_j(f_j)
    + \sum_{j>0} (\Psi^+)^j \Bigl( b_j + f_j \iota(\P)\Om + \bigl(
    \eta^{\mu\nu} \P_\mu\P_\nu \bigr) u_j \Bigr) + \s_1 \sum_{j>0}
    2 (\Psi^+)^j \E^+ u_j .
  \end{equation*}
  Applying \eqref{cocycle}, we see that
  \begin{equation*}
    \bigl( \P_\mu\Theta^\mu \bigr) \bigl( b_j + f_j \iota(\P)\Om +
    \bigl( \eta^{\mu\nu} \P_\mu\P_\nu \bigr) u_j \bigr) = 0 .
  \end{equation*}
  Hence there is a decomposition
  \begin{equation*}
    b_j + f_j \iota(\P)\Om + \bigl(\eta^{\mu\nu} \P_\mu\P_\nu
    \bigr) u_j = (j+1) \bigl( \P_\mu\Theta^\mu \bigr) v_j + g_j\,\Om ,
  \end{equation*}
  where $g_j\in\CO[\C,\Gamma]$. This allows us to rewrite the cocycle
  $z$ as
  \begin{equation*}
    z = b_0 + \sum_{j>0} \Bigl( A_j(f_j) + B_j(g_j) \Bigr)
    + \s_1 \sum_{j>0} \Bigl( 2 (\Psi^+)^j \E^+ u_j + (\Psi^+)^{j+1}
    v_j \Bigr) .
  \end{equation*}
  We conclude that a cohomology classes in $E_2=H^*(E_1,\s_1)$ take the
  general form
  \begin{equation*}
    z = [b_0] + \sum_{j>0} \Bigl( [A_j(f_j)] + [B_j(g_j)] \Bigr) ,
  \end{equation*}
  where $[b_0]$ is an element of the ring
  \begin{equation*}
    \CO[\Theta^\mu,\P_\mu,\C,\Gamma]/\bigl( \P_\mu\Theta^\mu ,
    \half\eta^{\mu\nu}\P_\mu\P_\nu \bigr) .
  \end{equation*}

  The differential $\s_2$ on $E_2$ is induced by the graded
  differential taking the following values on the generators of $E_1$:
  \begin{align*}
    \s_2 \E^+ &= 0 & \s_2 \Psi^+ &= 2 \E^+ \Gamma \\
    \s_2 \X^\mu &= - \eta^{\mu\nu} \P_\nu \C - \Theta^\mu \Gamma &
    \s_2 \Theta^\mu &= - \eta^{\mu\nu} \P_\nu \Gamma \\
    \s_2 \C &= \Gamma^2 & \s_2 \Gamma &= 0
  \end{align*}
  We have
  \begin{equation*}
    \s_2(B_j(f)) = A_j(\Gamma f) - \s_1 \left( \frac{(\Psi^+)^{j+1}
        \C \iota(\grad f)\Om}{j+1} \right) .
  \end{equation*}
  This establishes the formula $\s_2[B_j(f)]=[A_j(\Gamma f)]$ and
  $\s_2[A_j(f)]=0$ in $E_2$. Thus, the subcomplex of $E_2$ spanned by
  \begin{equation*}
    \bigl\{ [A_k(f)],[B_k(f)] \bigm| k>0,f\in\CO[\C,\Gamma]\bigr\}
  \end{equation*}
  contributes to $E_3$ a summand spanned by
  \begin{equation*}
    \bigl\{ [A_k(f)] \bigm| k>0,f\in\CO[\C] \bigr\} .
  \end{equation*}
  These elements of $E_3$ lift to the cocycles $\alpha_k(f)$ and
  $\zeta_k(f)$ in $\A$, showing that these classes in $E_3$ all
  survive to the last page $E_\infty$ of the spectral sequence.
  
  It remains to calculate the contribution of the elements of the form
  $[b_0]\in E_2$ to $E_3$. There is a quasi-isomorphism between the
  complexes
  \begin{equation*}
    \bigl( \CO[\Theta^\mu,\P_\mu,\C,\Gamma]/( \P_\mu\Theta^\mu ,
    \eta^{\mu\nu}\P_\mu\P_\nu) , \s_2 \bigr)
  \end{equation*}
  and $\R$, induced by sending the elements $\C$ and $\Gamma^2$ to
  zero and $\Gamma$ to $-\Gamma$. Thus the subspace of $E_2$
  consisting of elements of the form $[b_0]$ contributes a copy of the
  cohomology $H^*(\R)$ of $\R$ to $E_3$. These classes lift to the
  summand of $H^*(\A,\s)$ spanned by the classes
  \begin{equation*}
    \{\xi^0(u)\mid u\in H^0(\R) \} \oplus \{\xi^1(v)\mid v\in
    H^1(\R) \} ,
  \end{equation*}
  completing the proof of the theorem.
\end{proof}

We denote the transgressions of the cocycles $\alpha_k(f)$,
$\zeta_k(f)$, $\xi^0(u)$ and $\xi^1(v)$ obtained by applying the
vector field $\g$ by $\tilde\alpha_k(f)$, $\tilde\zeta_k(f)$,
$\tilde\xi^0(u)$ and $\tilde\xi^1(v)$. The following corollary follows
by applying the long-exact sequence \eqref{long}.
\begin{corollary}
  \begin{equation*}
    H^{-k}(\F,\s) =
    \begin{cases}
      \Bigl\{ \int \bigl( \alpha_k(f) + \zeta_k(g) +
      \tilde\alpha_k(\tilde{f}) + \tilde\zeta_k(\tilde{g}) \bigr) \,dt \,
      \Big| \, f,g,\tilde{f},\tilde{g} \in\CO \Bigr\} & k>1 , \\[10pt]
      \Bigl\{ \int \bigl( \tilde{\xi}^{-1}(u) + \alpha_1(f) + \zeta_1(g)
      + \tilde\alpha_1(\tilde{f}) + \tilde\zeta_1(\tilde{g}) \bigr) \,dt
      \,\Big|\, & \\
      \hfill u\in H^0(\R), f,g,\tilde{g} \in\CO , \tilde{f} \in
      \CO/\CC \Bigr\} & k=1 , \\[10pt]
      \Bigl\{ \int \bigl( \xi^0(u) + \tilde\xi^0(v) + \alpha_0(f) +
      \zeta_0(g) \bigr) \, dt \,\Big|\, & \\
      \hfill u \in H^0(\R) , v\in H^1(\R) , f \in \CO/\CC , g\in \CO
      \Bigr\} & k=0 , \\[10pt]
      \Bigl\{\int\xi^1(v)\,dt\,\Big|\, v\in H^1(\R) \Bigr\} & k=-1 ,
      \\[10pt]
      0 & k<-1 .
    \end{cases}
  \end{equation*}
\end{corollary}

In the remainder of this section, we calculate some examples of graded
Lie brackets among the cohomology groups $H^*(\F,\s)$. All of the
brackets among classes in the image of $H^*(\A,\s)\to H^*(\F,\s)$
vanish. As for brackets among the classes $\xi^i(u)$ and
$\tilde\xi^i(u)$, we have $\{ \xi^i(u) , \xi^j(v) \} = 0$, and
\begin{align*}
  \{ \tilde\xi^i(u) , \xi^j(v) \} &= \xi^{i+j}( \{u,v\} ) , &
  \{ \tilde\xi^i(u) , \tilde\xi^j(v) \}
  &= (-1)^{\pa(u)} \tilde\xi^{i+j}( \{u,v\} ) ,
\end{align*}
where
\begin{equation*}
  \{ u , v \} = (-1)^{\pa(u)} \frac{\p u}{\p x^\mu} \frac{\p v}{\p
    p_\mu} - (-1)^{\pa(u)} \frac{\p u}{\p p_\mu} \frac{\p v}{\p x^\mu}
  - \eta^{\mu\nu} \frac{\p u}{\p\theta^\mu} \frac{\p
    v}{\p\theta^\nu} .
\end{equation*}
Here, we understand $\xi^k(u)$ and $\tilde\xi^k(u)$ to be $0$ when
$k>1$.

We next calculate the bracket of $\tilde\alpha_k(f)$ with the classes
$\alpha_\ell(g)$, $\zeta_\ell(g)$, $\xi^i(u)$, taking the calculation
in the case $d=0$ as the model. We have
\begin{align*}
  \g(A_k(f)) &= (k+1) (\psi^+)^k \gamma^+ c f \Om \gamma^{-1} -
               (\psi^+)^{k+1} e f \Om \gamma^{-1} - (\psi^+)^{k+1} ( p^+ \cdot
               \grad f ) c \Om \gamma^{-1} \\
             &\quad + (\psi^+)^{k+1} c f \iota(\theta^+) \Om
               \gamma^{-1} + (-1)^d (\psi^+)^{k+1} c f \Om \psi
               \gamma^{-2} .
\end{align*}
Using this formula, it is not hard to show that
$\{\g(A_k(f)),\alpha_\ell(g)\} = \{\g(A_k(f)),\alpha_\ell(g)\} = 0$,
at least when $d>2$, and hence that
$\{\tilde\alpha_k(f),\alpha_\ell(g)\} =
\{\tilde\alpha_k(f),\zeta_\ell(g)\}=0$.
When $d\le2$, the brackets $\{\g(A_k(f)),\alpha_\ell(g)\}$ and
$\{\g(A_k(f)),\alpha_\ell(g)\}$ are coboundaries in the complex $\A$,
leading to the same conclusion.

We also see that
\begin{equation*}
  \{ \g(A_k(f)) , u \} = (-1)^d \, (\psi^+)^{k+1} c \left(
  \frac{\p f}{\p x^\mu} \frac{\p u}{\p p_\mu} \right) \Om \gamma^{-1} .
\end{equation*}
Let $\CE$ be the Euler vector field in the variables $p_\mu$:
\begin{equation*}
  \CE = p_\mu \frac{\p\ }{\p p_\mu} .
\end{equation*}
There exist functions
\begin{equation*}
  U^{\mu\nu}(x,p) = ( \CE + 1 )^{-1} \frac{\p^2u}{\p p_\mu\p
    p_\nu}(x,0,p)
\end{equation*}
such that
\begin{equation*}
  \frac{\p u}{\p p_\mu}(x,0,p) = \frac{\p u}{\p p_\mu}(x,0,0) + p_\nu
  U^{\mu\nu}(x,p) .
\end{equation*}
It follows that
\begin{multline*}
  \{ \g(A_k(f)) , u \} = (-1)^d \, A_k\left( \frac{\p f}{\p x^\mu}
    \frac{\p u}{\p p_\mu}(x,0,0) \right) + \frac{(-1)^d}{k+2} \, \s
  \Bigl( (\psi^+)^{k+2} c \iota( \grad f \cdot U ) \Om \gamma^{-1}
  \Bigr) \\
  + \frac{(-1)^d}{k+2} \, \Bigl( 2(k+2)(\psi^+)^{k+1} e^+ c \iota(
  \grad f \cdot U ) \Om \\ + (\psi^+)^{k+2} \iota( \grad f \cdot U ) \Om
  \gamma \bigr) + (\psi^+)^{k+2} c \iota( \grad f \cdot U ) \iota(p)
  \Om \Bigr) .
\end{multline*}
Suppose that $\Q u=0$, so that $u$ determines a cohomology class
$\xi^0(u)\in H^0(\F,\s)$. Applying the differential $\s$ to both
sides, we see that
\begin{multline*}
  \{ \tilde\alpha_k(f) , \xi^0(u) \} = (-1)^{d+1} \, \alpha_k\left(
    \frac{\p f}{\p x^\mu} \frac{\p u}{\p p_\mu}(x,0,0) \right) - \p \{
  A_k(f) , u \} \\
  - \frac{(-1)^d}{k+2} \, \s\Bigl( 2(k+2)(\psi^+)^{k+1} e^+ c \iota(
  \grad f \cdot U ) \Om \\ + (\psi^+)^{k+2} \iota( \grad f \cdot U )
  \Om \gamma \bigr) + (\psi^+)^{k+2} c \iota( \grad f \cdot U )
  \iota(p) \Om \Bigr) .
\end{multline*}
But $\{A_k(f),u\}$ vanishes, while the terms following it yield a
coboundary in $\A$. Thus we see that 
\begin{equation*}
  \{ \tilde\alpha_k(f) , \xi^0(u) \} = (-1)^{d+1} \, \alpha_k\left(
    \frac{\p f}{\p x^\mu} \frac{\p u}{\p p_\mu}(x,0,0) \right) .
\end{equation*}
A similar argument shows that
$\{\tilde\alpha_k(f) , \xi^1(v) \} = 0$.

\acknowledgments

I am grateful to Chris Hull for introducing me to the first-order
formalism of the spinning particle.

Igor Zavkhine informed me that an earlier version of this paper
contained an incorrect calculation of $H^{-1}(\F,\s)$. This revised
version of the paper corrects this, and calculates the cohomology in
positive degrees as well.

This research is partially supported by EPSRC Programme Grant
EP/K034456/1 ``New Geometric Structures from String Theory'' and
Collaboration Grant \#243025 of the Simons Foundation.

\end{document}